\documentclass[12pt]{amsart}
\usepackage{url}
\usepackage{amsmath,amsxtra,amssymb,latexsym,epsfig,amscd,amsthm,fancybox,epsfig}
\usepackage[mathscr]{eucal}
\usepackage{graphicx}
\usepackage{multicol,xcolor}
\usepackage{epsfig} 
\usepackage{epstopdf}
\usepackage{cases}
\usepackage{subfig}
\usepackage{color}
\usepackage{hyperref}

\hypersetup{
    colorlinks=true,
    linkcolor=blue,
    filecolor=magenta,
    urlcolor=cyan,
    citecolor=blue,
}

\usepackage{bigints}
\usepackage{enumitem}
\usepackage{natbib}

\newtheorem{thm}{Theorem}[section]
\newtheorem {asp}{Assumption}[section]

\newtheorem{rmk}{Remark}[section]

\theoremstyle{definition}
\newtheorem{defn}[thm]{Definition}

\theoremstyle{remark}

\newtheorem{example}{Example}[section]
\numberwithin{equation}{section}

\newcommand{\eps}{\varepsilon}

\newcommand{\EE}{\mathcal{E}}

\newcommand{\E}{\mathbb{E}}
\newcommand{\BE}{\mathbf{E}}
\newcommand{\BB}{\mathbf{B}}

\newcommand{\BX}{\mathbf{X}}
\newcommand{\bx}{\mathbf{x}}

\newcommand{\by}{\mathbf{y}}

\newcommand{\bc}{\mathbf{c}}

\newcommand{\N}{\mathbb{N}}
\newcommand{\CN}{\mathcal{N}}
\newcommand{\PP}{\mathbb{P}}

\newcommand{\R}{\mathbb{R}}

\numberwithin{equation}{section}

\newcommand{\bed}{\begin{displaymath}}
\newcommand{\eed}{\end{displaymath}}
\newcommand{\bea}{\bed\begin{array}{rl}}
\newcommand{\eea}{\end{array}\eed}

\newcommand{\barray}{\begin{array}{ll}}
\newcommand{\earray}{\end{array}}
\newcommand{\diag}{{\rm diag}}
\def\disp{\displaystyle}

\def\bar{\overline}

\def\a.s{\text{\;a.s.\;}}

\title[Competitive Exclusion Principle]{The competitive exclusion principle in stochastic environments}
\author[A. Hening]{Alexandru Hening }

\address{Sydney Mathematical Research Institute\\
University of Sydney\\
L4.42, Quadrangle A14\\
Sydney, NSW\\
Australia
}

\address{Department of Mathematics\\
Tufts University\\
Bromfield-Pearson Hall\\
503 Boston Avenue\\
Medford, MA 02155\\
United States
}
\email{alexandru.hening@tufts.edu}
\author[D. Nguyen]{Dang H. Nguyen }
\address{Department of Mathematics \\
University of Alabama\\
345 Gordon Palmer Hall\\
Box 870350 \\
Tuscaloosa, AL 35487-0350 \\
United States}
\email{dangnh.maths@gmail.com}

\keywords{Competitive exclusion; reversal; ergodicity; Lotka-Volterra; Lyapunov exponent; stochastic environment}
\subjclass[2010]{92D25, 37H15, 60H10, 60J05, 60J99}

\begin{document}
\maketitle
\begin{abstract}
In its simplest form, the competitive exclusion principle states that a number of species competing for a smaller number of resources cannot coexist. However, it has been observed empirically that in some settings it is possible to have coexistence. One example is Hutchinson's `paradox of the plankton'. This is an instance where a large number of phytoplankton species coexist while competing for a very limited number of resources. Both experimental and theoretical studies have shown that temporal fluctuations of the environment can facilitate coexistence for competing species. Hutchinson conjectured that one can get coexistence because nonequilibrium conditions would make it possible for different species to be favored by the environment at different times.

In this paper we show in various settings how a variable (stochastic) environment enables a set of competing species limited by a smaller number of resources or other density dependent factors to coexist. If the environmental fluctuations are modeled by white noise, and the per-capita growth rates of the competitors depend linearly on the resources, we prove that there is competitive exclusion. However, if either the dependence between the growth rates and the resources is not linear or the white noise term is nonlinear we show that coexistence on fewer resources than species is possible.
Even more surprisingly, if the temporal environmental variation comes from switching the environment at random times between a finite number of possible states, it is possible for all species to coexist even if the growth rates depend linearly on the resources. We show in an example (a variant of which first appeared in Benaim and Lobry '16) that, contrary to Hutchinson's explanation, one can switch between two environments in which the same species is favored and still get coexistence.

\end{abstract}
\section{Introduction}\label{s:intro}
The competitive exclusion principle \cite{V28, G32,H60, L70} loosely says that when multiple species compete with each other for the same resource, one competitor will win and drive all the others to extinction. Nevertheless, it has been observed in nature that multiple species can coexist despite limited resources. For example, phytoplankton species can coexist even though they all compete for a small number of resources. This apparent violation of the competitive exclusion principle has been called by Hutchinson `the paradox of the plankton' \cite{H61}. Hutchinson gave a possible explanation by arguing that variations of the environment can keep species away from the deterministic equilibria that are forecasted by the competitive exclusion principle.

\cite{H60} states the competitive exclusion principle as `complete competitors cannot coexist.' \cite{D84} quoting \cite{G32}, states it as `It is admitted that as a result of competition two similar species scarcely ever occupy similar niches, but displace each other in such a manner that each takes possession of certain peculiar kinds of food and modes of life in which it has an advantage over its competitor.' \cite{C00} defines the niche as `A species' niche is defined by the effect that a species has at each point in niche space, and by the response that a species has to each point.'

There has been continued debate regarding the competitive exclusion principle. Some have argued that the principle is a tautology or that since all species have finite population sizes they will eventually go extinct, therefore questioning the value of the principle. Analysing the competitive exclusion principle mathematically for a large class of models can guide us in this debate. Even though from a mathematical point of view, coexistence means that no species goes extinct in finite time, we will interpret this as providing evidence that no species will go extinct for a long period of time.
The first general deterministic framework for examining problems of competitive exclusion appeared in \cite{AM80}. This paper and the beautiful proofs from \cite{HS98} inspired us to look into how a variable environment enables a set of species limited by a smaller number of resurces or other density dependent factors to coexist.

It is well documented that one has to look carefully at both the biotic interactions and the environmental fluctuations when trying to determine criteria for the coexistence or extinction of species. Sometimes biotic effects can result in species going extinct. However, if one adds the effects of the environment, extinction might be reversed into coexistence. These phenomena have been seen in competitive settings as well as in settings where prey share common predators - see \cite{CW81, AHR98, H77}. In other instances, deterministic systems that coexist become extinct once one takes into account environmental fluctuations - see for example \cite{HN17}.
One successful way of analyzing the interplay between biotic interactions and temporal environmental variation is by modelling the populations as discrete or continuous-time Markov processes. The problem of coexistence or extinction then becomes equivalent to studying the asymptotic behaviour of these Markov processes. There are many different ways of modeling the random temporal environmental variation. One way that is widely used is adding white noise to the system and transforming differential equations into stochastic differential equations (SDE). However, for many systems, the randomness might not be best modelled by SDE \cite{T77}. Because of this, it is relevant to see how the long term fate of ecosystems is changed by different types of temporal environmental variation. The idea that extinction can be reversed, due to environmental fluctuations, into coexistence has been revisited many times since Hutchinson's explanation. A number of authors have shown that coexistence on fewer resources than species is possible as a result of interactions of species with temporal environmental variation ( \cite{CW81, C82, C94, LC16}). Our contribution to the literature of competitive exclusion is two-fold: 1) We develop powerful analytical methods for studying this question. 2) We prove general theoretical results and provide a series of new illuminating examples.

\section{The deterministic model}
Volterra's original model \cite{V28} assumed that the dynamics of $n$ competing species can be described using a system of ordinary differential equations (ODE). Most people who have studied the competitive exclusion principle mathematically have used ODE models. This is a key assumption and we will adhere to it in the current paper.
Suppose we have $n$ species $x_i, i=1,\dots,n$ and denote the density of species $i$ at time $t\geq 0$ by $x_i(t)$. Each species uses $m$ possible resources whose abundances are $R_j, j=1,\dots,m$. The resources themselves depend on the species densities, i.e. $R_j=R_j(\bx)$ is a function of the densities of the species $\bx(t)=(x_1(t),\dots,x_n(t)$. We assume that the per-capita growth rate of each species increases linearly with the amount of resources present. Based on the above, the dynamics of the $n$ species is given by
\begin{equation}\label{e:det}
dx_i(t) = x_i(t)\left(-\alpha_i+\sum_{j=1}^mb_{ij}R_j(\bx(t))\right)\,dt,\,~i=1,\dots,n
\end{equation}
where $-\alpha_i\leq 0$ is the rate of death in the absence of any resource, $R_k\geq 0$ is the abundance of the $k$th resource, and the coefficients $b_{ik}$ describe the efficiency of the $i$th species in using the $k$th resource. A key requirement is that the resources $R_k$ all eventually get exhausted. In mathematical terms this means that
\begin{equation}\label{e:exhaust}
R_k(\bx)=\bar R_{k}-F_k(\bx)
\end{equation}
where the $F_k$'s are unbounded positive functions of the population densities $x_i$ with $F_k(0,\dots,0)=0$.
This will make it impossible for the densities $x_i$ to grow indefinitely, and will be a standing assumption throughout the paper.

In the special case when the resources depend linearly on the densities, so that $F_k(\bx)=\sum_{i=1}^nx_ia_{ki}$ for constants $a_{ki}\geq 0$, equation \eqref{e:exhaust} becomes
\begin{equation}\label{e:linear}
R_k(\bx)=\bar R_{k}-\sum_{i=1}^nx_ia_{ki}
\end{equation}
and the system \eqref{e:det} is of Lotka--Volterra type.
The model given by \eqref{e:det} and \eqref{e:exhaust} is called by \cite{AM80} a \textit{linear abiotic resource model}. The linearity comes from \eqref{e:det} which intrinsically assumes that the per capita growth rates of the competing species are linear functions of the resource densities. The resources are abiotic because they regenerate according to the  algebraic equation \eqref{e:exhaust}, in contrast to being biotic and following systems of differential equations themselves.

The following result is a version of the competitive exclusion principle - see \cite{HS98} for an elegant proof.
\begin{thm}\label{CE_det}
Suppose $n>m$, the dynamics is given by \eqref{e:det}, and the resources eventually get exhausted. Then at least one species will go extinct.
\end{thm}

\begin{asp}\label{a:assumptions}
It is common to make the following assumptions when studying the competitive exclusion principle \cite{AM80}.
\begin{enumerate} [label=(\roman*)]
\item The populations are unstructured and as such the system can be fully described by the densities of the species.
\item The $n$ species interact with each other only through the resources. This way the growth rates of the species only depend on the resources $R_k, k=1,\dots,m$ and not directly on the densities $x_i, i=1,\dots,n$.
\item The resources all eventually get exhausted.
\item The growth rates of the species depend linearly on the resources that are available. Note that this is implicit in \eqref{e:det}.
\item The system is homogenous in space and the resources are uniform in quality.
\item There is no explicit time dependence in the interactions.
\item There is no random temporal environmental variation that can affect the resources and species.
\end{enumerate}
\end{asp}
When one or more of the assumptions (i)-(vi) are violated the coexistence of all species is possible. For example, if assumption (i) is violated it has been shown by \cite{HMS72} that two predators can coexist competing for the same prey if they eat different life stages (larval vs adult) of the prey. Similarly, two herbivores eating one plant can survive if they eat different parts of the plant. If (vi) is violated and the environment is time-varying it has been showcased by \cite{SL73, K74} that multiple species can coexist using a single resource. In the more general setting of competition, without specifying the dependence on resources, it has been shown by \cite{C80, DMS81} how deterministic temporal environmental variation can create a rescue effect and promote coexistence. If the linear dependence on the resources (iv) does not hold several results (\cite{K74b, Z75, AM76a, AM76b, KY77, MA77, AM80}) have shown that the coexistence of $n$ species competing for $m<n$ resources is possible.

\subsection{Competitive exclusion without Assumption \ref{a:assumptions} (iv)}
\cite{AM80} have relaxed the linearity constraint from Assumption \ref{a:assumptions} (iv) and studied general systems of $n$ species competing for $m$ abiotic resources. The dynamics is then given by
\begin{equation}\label{e:det_abiotic}
\begin{split}
\frac{dx_i(t)}{dt}&= x_iu_i(R_1,\dots,R_m),~i=1,\dots,n\\
R_j&= \bar R_j - F_j(x_1,\dots,x_n),~ j=1,\dots,m,
\end{split}
\end{equation}
where $u_i(R_1,\dots,R_m)$ is the per-capita growth rate of species $i$ when the resources are $(R_1,\dots,R_m)$.
The $R_j$'s are considered resources, so it is assumed that species growth rates will increase with resource availability, while resource densities will decrease with species densities. These conditions can be written as
\begin{equation}\label{e:mon}
\frac{\partial u_i}{\partial R_j}\geq 0~~\text{and}~~\frac{\partial F_j}{\partial x_i}\geq 0, ~i=1,\dots,n, ~j=1,\dots,m
\end{equation}
where the equalities hold if any only if species $i$ does not use resource $j$.

Volterra \cite{V28} proved that $n>1$ species cannot coexist if they compete for one abiotic resource. However, Volterra assumed as many others, that the $u_i$s from \eqref{e:det_abiotic} are linear, i.e. the growth rates depend linearly on the resources. If one assumes there is only one resource, surprisingly, the linearity assumption is not necessary. The conditions from \eqref{e:mon} are enough to force all but one species to go extinct. Only the species which can exist at the lowest level of available resource will persist and the following version of the competitive exclusion principle (see \cite{AM80}) holds.
\begin{thm}\label{CE_det2}
Suppose there are $n>1$ species competing for one abiotic resource $R$. If the dynamics is given by \eqref{e:det_abiotic} and the monotonicity conditions \eqref{e:mon} are satisfied then one species persists and all the others go extinct.
\end{thm}

We will study what happens when assumptions (i)-(vi) hold and assumption (vii) does not as well as how white noise interacts with the system when assumption (iv) fails.

\section{Stochastic coexistence theory}\label{s:stochastic}

In this section we describe some of the general stochastic coexistence theory that has been developed recently.
We start by defining what we mean by extinction and coexistence in the stochastic setting.  Assume $(\Omega, \mathcal{F}, \PP)$ is a probability space and let $(\BX(t))=(X_1(t),\dots,X_n(t))$ denote the densities of the $n$ species at time $t\geq 0$. We will assume that $(\Omega, \mathcal{F}, \PP)$ satisfies all the natural assumptions and that $\BX$ is a Markov process. We will denote by $\PP_\by(\cdot)=\PP(~\cdot~|~\BX(0)=\by)$ and $\E_\by[\cdot]=\E[~\cdot~|~\BX(0)=\by]$ the probability and expected value given that the process starts at $\BX(0)=\by$. Let $\partial\R_+^n=[0,\infty)^n\setminus (0,\infty)^n$ be the boundary of the positive orthant.
\begin{defn}\label{d:extinction}
Species $X_i$ goes \textbf{extinct} if for any initial species densities $\BX(0)\in(0,\infty)^n$ we have with probability $1$ that
\[
\lim_{t\to\infty} X_i(t)=0.
\]
We say that \textbf{at least one species goes extinct} if the process $\BX(t)$ converges to the boundary $\partial\R_+^n$ in the following sense: there exists $\alpha>0$ such that for any initial densities $\BX(0)\in (0,\infty)^n$ with probability $1$
\[
\limsup_{t\to\infty} \frac{\ln \left(d\left(\BX(t),\partial\R_+^n\right)\right)}{t}\leq -\alpha,
\]
where $d(\by,\partial\R_+^n)=\min\{y_1\dots,y_n\}$ is the distance from $\by$ to the boundary $  \partial\R_+^n$.
\end{defn}
\begin{defn}
The species $X_j$ is \textbf{persistent in probability} if
for every $\eps>0$, there exists $\delta>0$ such that for any $\BX(0)=\by \in (0,\infty)^n$ we have that
\[
\liminf_{t\to\infty}\PP_\by\left\{X_j(t)>\delta\right\} \geq 1-\eps.
\]
If all species $X_j$ for $j=1,\dots,n$ persist in probability we say the species \textbf{coexist}.
\end{defn}
This definition has first appeared in work by \cite{C78, C82}. There is a general theory of coexistence for deterministic models (\cite{H81, H84, HJ89, HS98, ST11}). It can be shown that a sufficient condition for coexistence is the existence of a fixed set of weights associated with the interacting populations, such that this weighted combination of the populations's invasion rates is positive for any invariant measure supported by the boundary (i.e. associated to a sub-collection of populations) -- see work by \cite{H81}. This coexistence theory has been generalized to stochastic difference equations in a compact state space (\cite{SBA11}), stochastic differential equations (\cite{SBA11,HN16}), and recently to general Markov processes (\cite{B18}).

The intuition behind the stochastic coexistence results is as follows. Let $\mu$ be an invariant probability measure of the process $\BX$ that is supported on the boundary $\partial\R_+^n$. Loosely speaking $\mu$ describes the coexistence of a sub-community of species, where at least one of the initial $n$ species is absent. If the process $\BX$ spends a lot of time close to (the support of) $\mu$ then it will get attracted or repelled in the $i$th direction according to the invasion rate $\Lambda_i(\mu)$.  This quantity can usually be computed by averaging some growth rates according to the measure $\mu$.
The invasion rate $\Lambda_i(\mu)$ quantifies how the $i$th species behaves when introduced at a low density into the sub-community supported by the measure $\mu$. If the invasion rate is positive, then the $i$th species tends to increase when rare, while if it is negative, the species tends to decrease when rare. We will use the following stochastic coexistence criterion for $n=2$ species.
\begin{thm}\label{t:pers}
Suppose species $X_1$ survives on its own and has the unique invariant measure $\mu_1$ on $(0,\infty)$. Similarly, assume species $X_2$ survives in the absence of $X_1$ and has the unique invariant measure $\mu_2$ on $(0,\infty)$. Assume furthermore that the invasion rates of the two species are strictly positive, i.e. $\Lambda_{x_1}:=\Lambda_1(\mu_2)>0$ and $\Lambda_{x_2}:=\Lambda_2(\mu_1)>0$. Then the two species coexist.
\end{thm}
A variant of this theorem appeared in work by \cite{CE89} in the setting of monotonic stochastic difference equations and then improved to more general stochastic difference equations by \cite{E89}. Moreover, \cite{CE89} develop specific conditions for coexistence in variable environments when there is but a single competitive factor, such as a single resource. This makes it a particularly relevant paper to our work.
For proofs of this theorem for stochastic differential equations see \cite{HN16}[Theorem 4.1 and Example 2.4] as well as \cite{B18}[Theorem 4.4 and Definition 4.3]. In the setting of PDMP see \cite{BL16} and \cite{B18}[Theorem 4.4 and Definition 4.3]. Other related persistence results have been shown by \cite{TG80, KO81, EHS15, SBA11, HN16, B18}.

\section{Stochastic differential equations}
\subsection{Growth rates depend linearly on resources}
One way of adding stochasticity to a deterministic system is based on the assumption that the environment mainly affects the vital rates of the populations. This way, the vital rates in an ODE (ordinary differential equation) model are replaced by their average values to which one adds a white noise fluctuation term; see the work by \cite{T77, B02, G88, ERSS13, SBA11, G84, HNY16} for more details. We note that just adding a stochastic fluctuating term to a deterministic model has some short comings because it does not usually explain how the biology of the species interacts with the environment. Instead, following the fundamental work by \cite{T77} we see the SDE models as ``approximations for more realistic, but often analytically intractable, models''. Moreover, as described by \cite{T77}, the It\^{o} interpretation (and not the Stratanovich one) of stochastic integration is the natural choice in the context of population dynamics. The general SDE model will be given by
\begin{equation}\label{e:SDE_gen}
dx_i(t) = x_i(t)f_i(\bx(t))\,dt + x_i(t) g_i(\bx(t))\,dE_i(t), ~i=1,\dots,n
\end{equation}
where $\BE(t)=(E_1(t),\dots, E_n(t))^T=\Gamma^\top\BB(t)$ for an $n\times n$ matrix
$\Gamma$ such that
$\Gamma^\top\Gamma=\Sigma=(\sigma_{ij})_{n\times n}$, $\BB(t)=(B_1(t),\dots, B_n(t))^T$ is a vector of independent standard Brownian motions, and $f_i, g_i:[0,\infty)^n\to\R$ are continuous functions that are continuously differentiable on $(0,\infty)^n$.
In this setting, if one has a subcommunity $M\subset\{1,\dots,n\}$ of species which has an invariant measure $\mu$ the invasion rate of the $i$th species is given by
\begin{equation}\label{e:lya}
 \Lambda_i(\mu):=\int_{\partial \R_+^n} \left(f_i(\bx)-\frac{g_{i}^2(\bx)\sigma_{ii}}{2}\right)\,d\mu.
\end{equation}
This expression can be seen as the average of the stochastic growth rate $f_i(\bx)-\frac{g_{i}^2(\bx)\sigma_{ii}}{2}$ with respect to the measure $\mu$.

We will first assume that the growth rates of the species depend linearly on the resources.
In this setting the system \eqref{e:det} becomes
\begin{equation}\label{e:SDE}
dx_i(t) = x_i(t)\left(-\alpha_i+\sum_{j=1}^mb_{ij}R_j(\bx(t))\right)\,dt + x_i(t) g_i(\bx(t))\,dE_i(t), ~i=1,\dots,n.
\end{equation}
Under appropriate smoothness and growth conditions, this system has unique solutions and $(0,\infty)^n$ is an invariant set for the dynamics, i.e. if the process starts in $(0,\infty)^n$ it will stay there forever.

The following stochastic version of the competitive exclusion principle holds.

\begin{thm}\label{t:CE_SDE}
Suppose $n$ species compete with each other according to \eqref{e:SDE}, the number of species is greater than the number of resources $n>m$, the resources depend on the species densities according to \eqref{e:exhaust} so that they eventually get exhausted and the random temporal environmental variation is linear, i.e. $g_i(\bx)=1$ for all $\bx\in [0,\infty)^n$ and all $i=1,\dots,n$. Then for any initial species densities $\bx(0)\in (0,\infty)^n$ with probability one at least one species will go extinct.
\end{thm}

We note that even though according to Theorem \ref{t:CE_SDE} white noise terms that are linear cannot facilitate coexistence, they can change which species go extinct and which persist as the next two-species example shows.
\begin{example}[Two dimensional Lotka--Volterra SDE]
Assume for simplicity we have two species $x_1, x_2$ competing for one resource $R$. Then if we assume the resource depends linearly on the species densities \eqref{e:linear} and we set $b_i:=b_{i1}$, $\mu_i=-\alpha_i+b_i \bar R, \beta_{ij}=b_ia_j$, and $g_i(\cdot)=1$ then the system \eqref{e:SDE} becomes
\begin{equation}\label{e:LV_SDE}
\begin{split}
dx_1(t)&=x_1(t)(\mu_1-\beta_{11}x_1(t)-\beta_{12}x_2(t))\,dt +x_1(t)dE_1(t)\\
dx_2(t)&=x_2(t)(\mu_2-\beta_{21}x_1(t)-\beta_{22}x_2(t))\,dt +x_2(t)dE_2(t).
\end{split}
\end{equation}
Suppose $\Sigma=\diag(\sigma_1^2,\sigma_2^2)$, $\mu_2-\frac{\sigma_2^2}{2}>0$, and $\mu_1-\frac{\sigma_1^2}{2}>0$ so that, according to the result by \cite{HN16}, none of the species go extinct on their own, as well as $\frac{b_1}{\mu_1}<\frac{b_2}{\mu_2}$ such that in the absence of random temporal environmental variation species $x_1$ dominates species $x_2$, i.e. $x_1$ persists while $x_2$ goes extinct. The following scenarios are possible (\cite{TG80, KO81, HN16, EHS15, HN17b})
\begin{itemize}
  \item If $\frac{b_1}{\mu_1}<\frac{b_2}{\mu_2} \frac{1-\frac{\sigma_1^2}{2\mu_1}}{1-\frac{\sigma_2^2}{2\mu_2}}$ then with probability one $x_1$ persists and $x_2$ goes extinct.
  \item If  $\frac{b_1}{\mu_1}>\frac{b_2}{\mu_2} \frac{1-\frac{\sigma_1^2}{2\mu_1}}{1-\frac{\sigma_2^2}{2\mu_2}}$ then with probability one $x_2$ persists and $x_1$ goes extinct.
\end{itemize}
The random temporal environmental variation acts on the dominance criteria according to the term $\frac{1-\frac{\sigma_1^2}{2\mu_1}}{1-\frac{\sigma_2^2}{2\mu_2}}$.
As a result, we can get \textit{reversal} in certain situations. Nevertheless, just as predicted by Theorem \ref{t:CE_SDE}, one species will always go extinct and the competitive exclusion principle holds.
\end{example}

This shows the competitive exclusion principle will hold when one models the environmental stochasticity by a white noise term of the form $x_i(t)dE_i(t)$ and if one assumes the growth rates of the species depend linearly on the resources. The linear random temporal environmental variation increases the expected resource level for each isolated species. The problem is that it also increases the death rates from $\alpha_i$ to $\alpha_i+\frac{\sigma_i^2}{2}$, therefore making coexistence impossible. A similar explanation was given by \cite{CH97} who studied the competition for a single resource in a variable environment and showed that a species might be subject to less competition when there is higher average mortality, but the higher average mortality counteracts the advantage of lower competition.

However, if the random temporal environmental variation term is not linear, the next result shows this need not be the case anymore.
\begin{thm}\label{t:CE_SDE3}
Assume that two species interact according to
\begin{equation}\label{e:SDE2}
\begin{split}
dx_1(t) &= x_1(t)\left(-\alpha_1+b_1R(\bx(t))\right)\,dt + x_1(t)\sqrt{\beta_1 x_1(t)} \,dB_1(t)\\
dx_2(t) &= x_2(t)\left(-\alpha_2+b_2R(\bx(t))\right)\,dt + x_2(t)\sqrt{\beta_2 x_2(t)} \,dB_2(t)
\end{split}
\end{equation}
and the resource $R$ depends linearly on the species densities, i.e. \eqref{e:linear} holds.
\begin{itemize}
  \item[i)] Suppose that $b_i\bar R>\alpha_i, i=1,2$. Then each species $i\in \{1,2\}$ can survive on its own and has a unique invariant probability measure $\mu_i$ on $(0,\infty)$.
  \item[ii)] Suppose in addition that the coefficients are such that the invasion rates are strictly positive, i.e. $$\Lambda_{x_2}=\int_0^\infty \left(-\alpha_2+b_2(\bar R-a_1 x)\right)\mu_1(dx)
=(b_2\bar R-\alpha_2)- b_2a_1 \dfrac{b_1\bar R-\alpha_1}{b_1a_1+\beta_1}>0
$$
and $$
\Lambda_{x_1}=\int_0^\infty \left(-\alpha_1+b_1(\bar R-a_2 x)\right)\mu_2(dx)=(b_1\bar R-\alpha_1)- b_1a_2 \dfrac{b_2\bar R-\alpha_2}{b_2a_2+\beta_2}>0.$$ Then the two species coexist.
\end{itemize}
\end{thm}
\begin{rmk}
We note that, as remarked by Peter Chesson, it is not clear how to interpret this result biologically. This is due to the fact that the $x^{3/2} dB$ noise term has the effect of strongly increasing the intraspecific density-dependence without revealing a biologically coherent mechanism.
One way of looking into this mechanism would be the following.
By \cite{T77} an It\^{o} stochastic differential equation of the form
\begin{equation}\label{e:genSDE}
dX_t = X_tf(X_t)\,dt + X_tg(X_t)\,dB_t
\end{equation}
can be seen as a scaling limit $N\to \infty$ of $X_N(t)=X_{\lfloor Nt\rfloor}^{(N)}$ where $X_n^{(N)}$ is the solution of the stochastic difference equations
\begin{equation}\label{e:sdiff}
X_{n+1}^{(N)}-X_n^{(N)} = f_N\left(X_n^{(N)}\right)\frac{1}{N} + g_N\left(X_n^{(N)}\right)\frac{\eta_{n+1}^{(N)}}{\sqrt N}
\end{equation}
where for each $N$, $\left(\eta_k^{(N)}\right)_{k\in \N}$ is a sequence of i.i.d random variables with mean $0$ and variance $1$, $f_N(x)$ and $g_N(x)$ agree with $xf(x)$ and $xg(x)$ for $x$ less than some large value $C_N$, and $C_N\to\infty$ as $N\to\infty$. As a consequence, one can interpret \eqref{e:genSDE} by looking at \eqref{e:sdiff}. 
\end{rmk}

The nonlinear random temporal environmental variation terms create a nonlinearity when computing the expected values of the resource when each species is on its own. This breaks the symmetry when computing the invasion rates and allows to have both invasion rates be strictly positive. One example of parameters for which we get coexistence is presented in Figure \ref{fig_sde}.

\begin{figure}[h]
\centering\includegraphics[width=0.6\linewidth]{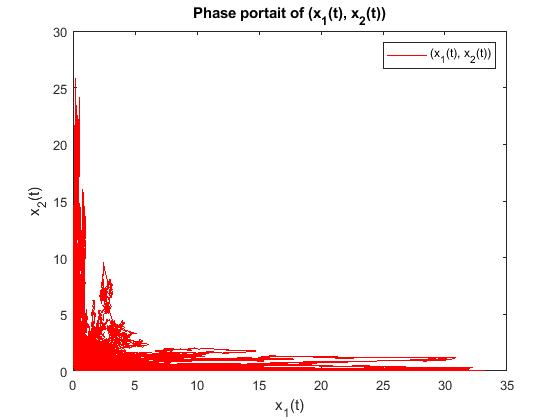}
 \caption{Example showing the coexistence of the species $x_1$ and $x_2$ in the SDE setting from Theorem \ref{t:CE_SDE3}. The paramters are $\alpha_1=0.5, \beta_1=4, \alpha_2=0.6, \beta_2=4, R=3, a_1=a_2=1,
$ and the invasion rates $\Lambda_{x_2}=1.9, \Lambda_{x_1}=2.02$. }
\label{fig_sde}
\end{figure}

\subsection{Non-linear dependence on the resources.}

If the assumption that the dependence of the per-capita growth rates on the resources is linear is dropped like in Theorem \ref{CE_det2} and the random temporal environmental variation is modeled by linear white noise multiple species can coexist while competing for one resource. The nonlinear dependence on resources falls under the coexistence mechanism described by \textit{relative nonlinearity}. This is a mechanism which makes coexistence possible via the different ways in which species use the available resources (\cite {AM80}). In stochastic environments this effect has been studied in discrete time by \cite{C94, YC15}.

\begin{thm}\label{t:CE_SDE2}

Suppose the dynamics of the two species is given by
\begin{equation}\label{e:sde_counter}
\begin{split}
dx_1(t)&=x_1(t)(-\alpha_1+f(\bar R-a_1x_1(t)-a_2x_2(t)))\,dt +\sigma_1x_1dB_1(t)\\
dx_2(t)&=x_2(t)(-\alpha_2+(\bar R-a_1x_1(t)-a_2x_2(t)))\,dt
\end{split}
\end{equation}
where $f$ is a continuously differentiable Lipschitz function satisfying $\lim_{x\to-\infty}f(x)=-\infty$,  $\dfrac{df(x)}{dx}>0, \dfrac{d^2f(x)}{dx^2}\leq 0$ for all $x\in\R$ and $\dfrac{d^2f(x)}{dx^2}<0$ for $x$ in some subinterval of $\left(-\infty,\frac{\bar R}{a_1}\right)$.
Let $a_1, a_2, \sigma_1, \alpha_1,\sigma_1, \bar R$ be any fixed positive constants satisfying $f(\bar R)>\alpha_1+\dfrac{\sigma_1^2}{2}$.
Then there exists an interval $(c_0,c_1)\subset (0,\infty)$ such that
the two species coexist for all $\alpha_2\in(c_0,c_1)$.
\end{thm}
\begin{rmk}\label{r:nonlinear}
A particular example is the following. Let $\bar R=5, a_1=a_2=2; \sigma_1=1, \alpha_1=0.5, \alpha_2=0.4$ and the function $f=f^*$ for
$$
f^*(x) = \left\{
        \begin{array}{ll}
            \ln\left( x+3\right) & \quad x\geq -2 \\
            x+2 & \quad x \leq -2.
        \end{array}
    \right.
$$
Then the two species modelled by \eqref{e:sde_counter} coexist.
\end{rmk}
We prove this is true in Appendix \ref{a:sde}. An example of two-species dynamics which satisfies this, is

The intuition is as follows: Consider \eqref{e:sde_counter} for an arbitrary function $f$. One can show that if one considers the species $x_1$ in the absence of species $x_2$, i.e.
\[
dx(t)=x(t)(-\alpha_1+f(\bar R-a_1x(t)))\,dt +\sigma_1xdB_1(t).
\]
then, under certain conditions, the process $(x(t))$ has a unique stationary distribution $\mu$ on $(0,\infty)$. Ergodic theory then implies $\alpha_1+\frac{\sigma_1^2}{2}=\int_0^\infty f(R-a_1x)\mu(dx)$.
If the function $f$ is concave then by Jensen's inequality and taking inverses
\begin{equation}\label{e:exp_R}
\int_0^\infty(R-a_1x)\mu(dx)>  f^{-1}\left(\alpha_1+\frac{\sigma_1^2}{2}\right).
\end{equation}
The concavity of $f$ increases the expected value of the resource $R$. However, in the deterministic setting or if $f$ is linear and there is no random temporal environmental variation, one would have equality
\[
\int_0^\infty(R-a_1x)\mu(dx)= f^{-1}\left(\alpha_1+\frac{\sigma_1^2}{2}\right).
\]
This is the main intuition behind the counterexample \eqref{e:sde_counter}. Because $f^*$ is concave, we can see that there will be by \eqref{e:exp_R} an increase in the expected value of the resource. This will in turn make coexistence possible. If $f$ is linear or $\sigma_1=0$, i.e. the system is deterministic, this cannot happen, and we always have competitive exclusion by Theorems \ref{CE_det2} or \ref{t:CE_SDE}.
\begin{figure}[h]
\centering\includegraphics[width=0.6\linewidth]{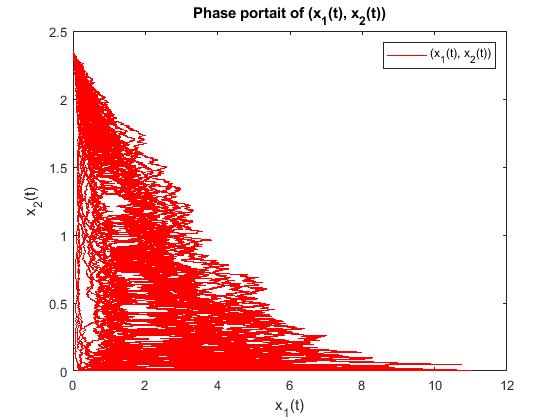}
 \caption{Example showing the coexistence of the species $x_1$ and $x_2$ in the SDE setting from Remark \ref{r:nonlinear}. The invasion rates are $\Lambda_{x_2}=0.192$ and $\Lambda_{x_1}=0.147$. }
\label{fig3}
\end{figure}

\section{Piecewise deterministic Markov processes}\label{s:pdmp}

The basic intuition behind piecewise deterministic Markov processes (PDMP) is that due to different environmental conditions, the way species interact changes. For example, in \cite{TL16}, it has been showcased that the predation behavior can vary with the environmental conditions and therefore change predator-prey cycles.
Since the environment is random, its changes (or switches) cannot be predicted in a deterministic way. For a PDMP, the process follows a deterministic system of differential equations for a random time, after which the environment changes, and the process switches to a different set of ordinary differential equations (ODE), follows the dynamics given by this ODE for a random time and then the procedure gets repeated. This class of Markov processes was first introduced in the seminal paper of Davis \cite{D84} and has been used in various biological settings \cite{CDGMMY17}, from population dynamics \cite{BL16, HS17, B18, DN11, DN14} to studies of the cell cycle \cite{LM99}, neurobiology \cite{DL17}, cell population models \cite{BDMT11}, gene expression \cite{YZLM14} and multiscale chemical reaction network models \cite{HGK15}.

Suppose $(r(t))$ is a process taking values in the finite state space $\CN=\{1,\dots,N\}.$ This process keeps track of the environment, so if $r(t)=i\in\CN$ this means that at time $t$ the dynamics takes place in environment $i$. Once one knows in which environment the system is, the dynamics are given by a system of ODE. The PDMP version of \eqref{e:det} therefore is

\begin{equation}\label{e1-pdm}
dx_i(t) = x_i(t)\left(-\alpha_i(r(t))+\sum_{j=1}^mb_{ij}(r(t))R_j(\bx(t),r(t))\right)\,dt.
\end{equation}

In order to have a well-defined system one has to specify the switching-mechanism, e.g. the dynamics of the process $(r(t))$.
Suppose that the switching intensity of $r(t)$ is given as follows
\begin{equation}\label{e:tran}\begin{array}{ll}
&\disp \PP\{r(t+\Delta)=j~|~r(t)=i, \bx(s),r(s), s\leq t\}=q_{ij}\Delta+o(\Delta) \text{ if } i\ne j \
\hbox{ and }\\
&\disp \PP\{r(t+\Delta)=i~|~r(t)=i, \bx(s),r(s), s\leq t\}=1+q_{ii}\Delta+o(\Delta)
\end{array}
\end{equation}
where $q_{ii}:=-\sum_{j\ne i}q_{ij}$. Here, we assume that the the matrix $Q=(q_{ij})_{N \times N}$ is irreducible. It is well-known that a process $(\BX(t),r(t))$ satisfying \eqref{e1-pdm} and \eqref{e:tran}
is a strong Markov process \cite{D84} while $(r(t))$ is a continuous-time Markov chain that has a unique invariant probability measure $\nu$ on $\CN$.

We define for $u\in\CN$ the $u$th environment, $\mathcal{E}_u$. This is the deterministic setting where we follow \eqref{e1-pdm} with $r(t)=u$ for all $t\geq 0$, i.e.
\begin{equation}\label{e:envu}
dx_i^u(t) = x_i^u(t)\left(-\alpha_i(u)+\sum_{j=1}^mb_{ij}(u)R_j(\bx^u(t),u)\right)\,dt.
\end{equation}
The dynamics of the switched system can be constructed as follows: We follow the dynamics of $\mathcal{E}_u$ and switch between environments $\EE_u$ and $\EE_v$ at the rate $q_{uv}$.
It is interesting to note that in the limit case where the switching between the different states is fast, the dynamics can be approximated (\cite{CDGMMY17, BS19}) by the `mixed' deterministic dynamics
\begin{equation}\label{e:mix}
d\bar x_i(t) = \bar x_i(t)\sum_{u\in\CN}\nu_u\left(-\alpha_i(u)+\sum_{j=1}^mb_{ij}(u)R_j(\bar \bx(t),u)\right)\,dt.
\end{equation}

If the number of species is strictly greater than the number of resources, $n>m$, for any $u\in\{1,\dots,N\}$ the system
\begin{equation}\label{e:systems}
\sum_{i=1}^nc_{i}(u) b_{ij}(u)=0, ~j=1,\dots,m
\end{equation}
admits a nontrivial solution $(c_1(u), \dots, c_n(u))$. We can prove the following PDMP version of the competitive exclusion principle. A related result has been stated informally in the discrete-time work by \cite{CH97}.

\begin{thm}\label{t:CE_PDMP}
Assume the dynamics of $n$ competing species is given by \eqref{e1-pdm} and \eqref{e:tran}, there are fewer resources than species $m<n$, and all resources eventually get exhausted. In addition, suppose that
\begin{equation*}
\lim_{\|\bx\|\to\infty} \left(-\alpha_i(u)+\sum_{j=1}^mb_{ij}(u)R_j(\bx,u) \right)<0, i=1,\dots, n, u=1,\dots,N.
\end{equation*}
and there exists a non-zero vector $(c_1,\dots,c_n)$ that is simultaneously a solution of the linear systems \eqref{e:systems} for all $u\in\{1,\dots,N\}$. Then, with probability $1$, at least one species goes extinct except possibly for the critical case when
\begin{equation*}
\sum_{i=1}^n c_i\sum_{k=1}^N\alpha_i(k)\nu_k=0,
\end{equation*}
where $(\nu_k)_{k\in\CN}$ is the invariant probability measure of the Markov chain $(r(t))$..
\end{thm}

This shows that competitive exclusion holds if there is some kind of `uniformity' of solutions of \eqref{e:systems} in all the different environments. However, the next example shows coexistence on fewer resources than species is possible for PDMP.

Suppose we have two species, two environments, one resource and the dependence of the resource on the population densities is linear, i.e. \eqref{e:linear} holds. In environment $\EE_u, u\in\{1,2\}$ the system is modelled by the ODE
\[
dx_i^u(t) = x_i^u(t)\left(-\alpha_i(u)+b_{i}(u) \left[\bar R(u)-\sum_{i=1}^2 x^u_i(t)a_{i}(u)\right]\right)\,dt
\]
and therefore the switched system is given by
\begin{equation}\label{e:switch}
dx_i(t) = x_i(t)\left(-\alpha_i(r(t))+b_{i}(r(t)) \left[\bar R(r(t))-\sum_{i=1}^2 x_i(t)a_{i}(r(t))\right]\right)\,dt.
\end{equation}

If we define $\mu_i(u)=-\alpha_i(u)+b_i(u)\bar R(u), \beta_{ij}(u)=b_i(u)a_j(u)$ we get the well known two-dimensional competitive Lotka--Volterra system
\begin{equation}\label{e:2d}
\begin{split}
dx_1(t)&=x_1(t)\mu_1(r(t))\left(1-\frac{\beta_{11}(r(t))}{\mu_1(r(t))}x_1(t)-\frac{\beta_{12}(r(t))}{\mu_1(r(t))}x_2(t)\right)\,dt\\
dx_2(t)&=x_2(t)\mu_2(r(t))\left(1-\frac{\beta_{21}(r(t))}{\mu_2(r(t))}x_1(t)-\frac{\beta_{22}(r(t))}{\mu_2(r(t))}x_2(t)\right)\,dt.
\end{split}
\end{equation}

By the deterministic competitive exclusion principle from Theorem \ref{CE_det} we know that in each environment $\EE_u, u\in\{1,2\}$ one species is dominant and drives the other one extinct.
\begin{thm}\label{t:CE_PDMP2}
Suppose two species compete according to \eqref{e:switch}. There exist environments $\EE_1, \EE_2$ for which the maximal resource is equal $\bar R(1)=\bar R(2)$ such that
\begin{enumerate}
\item in both environments $\EE_1, \EE_2$ species $x_1$ persists and species $x_2$ goes extinct, or
\item in environment $\EE_1$ species $x_1$ persists and species $x_2$ goes extinct while in environment $\EE_2$ the reverse happens and $x_1$ goes extinct while $x_2$ persists,
\end{enumerate}
and rates $q_{12}, q_{21}>0$ such that the process $\bx(t)$ modelled by \eqref{e:switch} converges to a unique invariant measure supported on a compact subset $K$ of the positive orthant $(0,\infty)^2$. In particular, with probability $1$ the two species coexist, and the competitive exclusion principle does not hold.
\end{thm}

\begin{rmk}
Here, the results of \cite{C80, DMS81} are related, even though deterministic. \cite{LC16} investigate a version of this model in which the environment can be deterministic or stochastic, with the sole requirement of stationarity of the environment. Their work shows mechanistically and biologically how coexistence occurs. They consider explicit resource dynamics, but in the limit of fast resource dynamics, their model becomes a version of our model.
\end{rmk}

We emphasize that the maximal resource does not have to change with the environment - in the above example the maximal resources in the two environments $\EE_1$ and $\EE_2$ are equal.
Two examples of systems satisfying Theorem \ref{t:CE_PDMP2} are given in Figures \ref{fig1} and \ref{figsw}. For the environments given by the coefficients from Figure \ref{fig1} one notes that species $x_1$ persists and $x_2$ goes extinct in $\EE_1$ while the reverse happens in environment $\EE_2$.
Even more surprisingly, for the environments given by the coefficients from Figure \ref{figsw} species $x_1$ persists and $x_2$ goes extinct in both environments.
By spending time in both environments there is a rescue effect which forces both species to persist. We note that Theorem \ref{t:CE_PDMP2} can be proved using results by \cite{BL16}.

\begin{figure}[h]
\centering\includegraphics[width=0.6\linewidth]{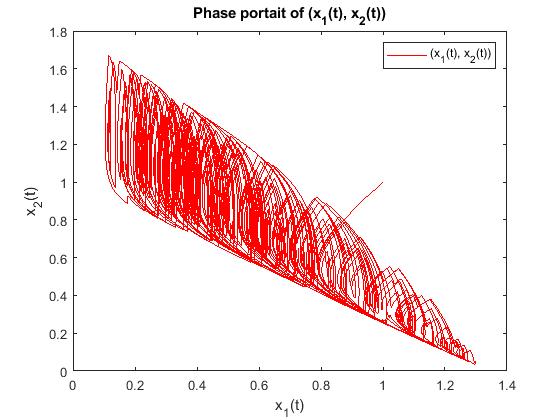}
 \caption{Example showing the coexistence of the species $x_1$ and $x_2$ when one switches between two environments. Species $x_1$ persists in $\EE_1$ and goes extinct in $\EE_2$ while the reverse happens for species $x_2$. The constants are $\alpha_1(1)=\alpha_1(2)=0.66, \alpha_2(1)=\alpha_2(2)=1, \bar R(1)= \bar R(2)=2, a_1=1, a_2=1, b_1(1)=b_1(2)=1, b_2(1)=1, b_2(2)=5, q_{12}=1, q_{21}=5$. The invasion rates of the two species are $\Lambda_{x_1}\approx0.137,
\Lambda_{x_2}\approx0.1$}
\label{fig1}
\end{figure}
\begin{figure}[h]
\centering\includegraphics[width=0.6\linewidth]{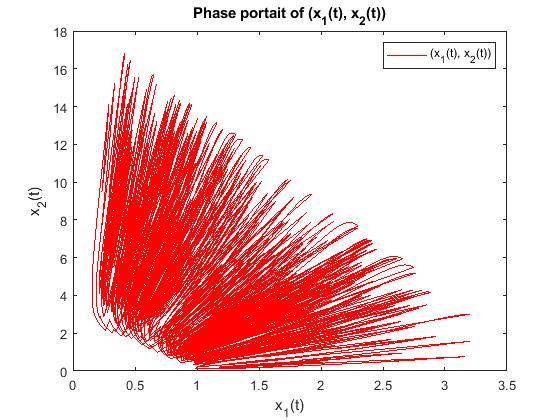}
 \caption{Example showing the coexistence of the species $x_1$ and $x_2$ when one switches between two environments in both of which species $x_1$ dominates and drives species $x_2$ extinct. The constants are $\alpha_1(1)=3, \alpha_1(2)=0.2, \alpha_2(1)=3.5, \alpha_2(2)=0.8, \bar R(1)= \bar R(2)=4, a_1=1, a_2=0.2, b_1(1)=1, b_1(2)=2, b_2(1)=1, b_2(2)=4, q_{12}=1, q_{21}=5$. The invasion rates of the two species are $\Lambda_{x_1}\approx0.00531,
\Lambda_{x_2}\approx0.00519$}
\label{figsw}
\end{figure}

\section{Discussion}\label{s:disc}
We have analyzed how environmental stochasticity influences the coexistence of $n$ species competing for $m<n$ abiotic resources. The assumptions we make are the ones that are common throughout the literature: the populations are unstructured, the species compete through the resources which eventually get exhausted, there is no explicit time dependence in the interactions and there is no environmental stochasticity. Another common assumption is that the per-capita growth rates of the species depend linearly on the resources.
There are several papers which have looked at related problems.
The first of these ( \cite{C94}) develops a general theory of coexistence in a variable environment. \cite{C09} gives a simpler presentation of the coexistence theory. \cite{K10} studies coexistence with the environment jumping between discrete states, which is an issue taken up in the current manuscript. \cite{LC16} is a detailed discussion of Hutchinson's paradox of the plankton. Finally, \cite{C18} is relevant as an overall review. We note that in most of the stochastic results one of the main assumptions is that the random temporal fluctuations are small. In our analysis, especially in the setting from Section \ref{s:pdmp}, this is not true anymore - the random fluctuations can, and will be, large. The small effects approximations in earlier papers have provided explicit formulae for species coexistence in a number of useful cases (\cite{C94}). In our work, explicit coexistence criteria are not as readily available due to the more complicated underlying mathematical structure.

Following \cite{C00} we note that the response of a species to random environmental fluctuations is part of the niche of the species. The coexistence ideas in the current paper also involve niche differences. We are able to show that in certain situations coexistence on fewer resources than species is possible as a result of the species interacting with the random environment.

In the setting of stochastic differential equations, if we assume that the per-capita growth rates of the species depend linearly on the resources and the white noise term is linear, we prove the stochastic analogue of the competitive exclusion principle holds: for any initial starting densities, at least one species will go extinct with probability $1$. The random temporal environmental variation can change which species persist and which go extinct as well as drive more species towards extinction. This is in line with \cite{HN17, HN17b} where we show that white noise of the form $x_i(t)dE_i(t)$ makes the coexistence of species in a Lotka--Volterra food chain less likely. In a sense, this type of white noise is, on average, detrimental to the ecosystem if the interaction between the species is linear enough.

However, if one drops the assumption that the per-capita growth rates of the species depend linearly on the resources, then coexistence on fewer resources than species is possible. We exhibit an example of two species competing for one resource where the species coexist because of the linear white noise. More specifically, we look at the interaction modelled by
\[
\begin{split}
dx_1(t)&=x_1(t)(-\alpha_1+f(\bar R-a_1x_1(t)-a_2x_2(t)))\,dt +\sigma_1x_1dB_1(t)\\
dx_2(t)&=x_2(t)(-\alpha_2+(\bar R-a_1x_1(t)-a_2x_2(t)))\,dt.
\end{split}
\]
 The combination of random temporal environmental variation and non-linear dependence on the resources make it possible for one species to get an increased expected value of the resource. This will in turn make it possible for the two species to coexist. To glean more information, we look at the invasion rates of the two species, namely
\[
\Lambda_{x_2}=\eps_0+   f^{-1}\left(\alpha_1+\dfrac{\sigma_1^2}{2}\right) - \alpha_2,
\]
and
\[
\Lambda_{x_1}=-\alpha_1-\frac{\sigma_1^2}{2}+f(\alpha_2).
\]
The constant $\eps_0$ is defined in the Appendix via equation \eqref{sde-e3} and involves the invariant probability distribution $\mu$ of species $x_1$.
We note that the invasion rates are nonlinear functions of the death rates of the species and the variance of the random temporal fluctuations.
This shows that the variance of the noise increases the invasion rate of $x_2$ and decreases the invasion rate of $x_1$, creating a type of relative nonlinearity (\cite{C94, YC15}). This well known mechanism, in turn, promotes coexistence. The conditions for coexistence in this setting are given by
$$f(\bar R)>\alpha_1+\dfrac{\sigma_1^2}{2}$$
and
$$\alpha_2\in\left(f^{-1}\left(\alpha_1+\dfrac{\sigma_1^2}{2}\right), \bar R\wedge\left(f^{-1}\left(\alpha_1+\dfrac{\sigma_1^2}{2}\right)+\eps_0\right)\right).$$
The first species needs to be efficient enough at using the resource, while the death rate of the second species cannot be too low, as that would make the invasion of species $1$ negative, nor can it be too high, as that would make its own invasion rate negative.

If instead, we drop the condition that the random temporal environmental variation term is linear, we construct an example of two species
\[\begin{split}
dx_1(t) &= x_1(t)\left(-\alpha_1+b_1R(\bx(t))\right)\,dt + x_1(t)\sqrt{\beta_1 x_1(t)} \,dB_1(t)\\
dx_2(t) &= x_2(t)\left(-\alpha_2+b_2R(\bx(t))\right)\,dt + x_2(t)\sqrt{\beta_2 x_2(t)} \,dB_2(t)
\end{split}
\]
where the random temporal environmental variation looks like $x_i(t)^{3/2}dE_i(t)$ and the per-capita growth rates of the species depend linearly on the resources in which the two species coexist. The invasion rates in this setting are
 $$\Lambda_{x_2}
=(b_2\bar R-\alpha_2)- b_2a_1 \dfrac{b_1\bar R-\alpha_1}{b_1a_1+\beta_1}>0
$$
and $$
\Lambda_{x_1}=\int_0^\infty \left(-\alpha_1+b_1(\bar R-a_2 x)\right)\mu_2(dx)=(b_1\bar R-\alpha_1)- b_1a_2 \dfrac{b_2\bar R-\alpha_2}{b_2a_2+\beta_2}>0.$$
Here there are two mechanisms that promote coexistence: relative nonlinearity (the invasion rates are nonlinear functions of the competition parameters) and the storage effect (density dependence of the covariance between the environment and competition).

The above two examples show that in order to have competitive exclusion it is key to assume both that the growth rates depend linearly on the resources and that the white noise term is linear. If either one of these assumptions is violated we are able to give examples of two species that compete for one resource and coexist, therefore violating the competitive exclusion principle. Nonlinear terms facilitate the coexistence of species. Since there is no reason one should assume the interactions or the random temporal environmental variation terms in nature are linear, this can possibly explain coexistence in some empirical settings.

The second type of random temporal environmental variation we analyze is coming from switching the environment between a finite number of states at random times, and following a system of ODE while being in a fixed environment. This is related to the concept of seasonal forcing, i.e. the aspect of nonequilibrium dynamics that looks at the temporal variation of the parameters of a model during the year. This has been studied extensively (\cite{H80, RMK93, KS99, LK01}) and was shown to have significant impacts on competitive, predator-prey, epidemic and other systems. However, much of the work in this area has been done using simulation or approximation techniques and did not involve any random temporal variation. We present some theoretical findings regarding the coexistence of competitors, in the more natural setting when the forcing is random.
We prove that if the different environments are uniform, in the sense that there exists a solution $(c_1,\dots,c_n)$ that solves the system
\begin{equation}\label{e:env}
\sum_{i=1}^nc_{i} b_{ij}(u)=0, ~j=1,\dots,m,
\end{equation}
simultaneously in an all environments then the competitive exclusion principle holds. If this condition does not hold, we construct an example, based on the work by \cite{BL16}, with two species $x_1,x_2$ competing for one resource, and two environments $\EE_1$ and $\EE_2$ such that in the switched system the two species coexist. We note that in this setting we do have that the growth rates of the species depend linearly on the resource. This example is interesting as it relates to Hutchinson's explanation of why environmental fluctuations can favor different species at different times and thus facilitate coexistence \cite{H61, LC16}. We are able to find environments $\EE_1$ and $\EE_2$ such that without the switching species $x_1$ persists and species $x_2$ goes extinct in both environments. However, once we switch randomly between the environments we get coexistence (see Figure \ref{figsw}). This implies the surprising result that species can coexist even if one species is unfavored at all times, in all environments. We conjecture that in general if one has $k$ environments and $m$ resources, the coexistence of $n\leq mk$ species will be possible if the environments are different, i.e. there is no solution of \eqref{e:env} that is independent of the environments. If the environments are different enough, each environment creates $m$ niches for the species. However, if the environments are too similar, i.e. \eqref{e:env} holds then coexistence is not possible.

Our analysis shows that different types of random temporal environmental variation interact differently with competitive exclusion according to whether the growth rates depend linearly on the resources or not. As long as the random temporal environmental variation is `smooth' and `linear' and changes the dynamics in a continuous way and the growth rates are linear in the resources, the competitive exclusion principle will hold. One needs nonlinear continuous random temporal environmental variation, `discontinuous' random temporal environmental variation that abruptly changes the dynamics of the system, or a nonlinearity in the dependence of the per-capita growth rates on the resources in order to facilitate coexistence.

{\bf Acknowledgments.}  We thank Jim Cushing and Simon Levin for their helpful suggestions. The manuscript has improved significantly due to the comments of Peter Chesson and one anonymous referee. The authors have been in part supported by the NSF through the grants DMS 1853463 (A. Hening) and DMS 1853467 (D. Nguyen). Part of this work has been done while AH was visiting the University of Sydney through an international visitor program fellowship.
\bibliographystyle{agsm}

\bibliography{LV}

\appendix
\section{Proof of Theorem \ref{t:CE_SDE}}
If the number of species is strictly greater than the number of resources, $n>m$, the system
\begin{equation}\label{e:system}
\sum_{i=1}^nc_{i} b_{ij}=0, ~j=1,\dots,m
\end{equation}
admits a nontrivial solution $(c_1, \dots, c_n)$.
\begin{thm}
Assume that $\lim_{\|\bx\|\to\infty} R_j(\bx)=-\infty, j=1,\dots,m$.
Suppose further that $n$ species interact according to \eqref{e:SDE}, the number of species is greater than the number of resources $n>m$ and the resources depend on the species densities according to \eqref{e:exhaust} so that they eventually get exhausted.
Suppose further that $g_i(\bx)=1$ and
\[
0<r_m\leq \liminf_{\|\bx\| \to\infty} \frac{|R_{j}(\bx)|}{|R_1(\bx)|}\leq r^M<\infty,
\]
for $j=1,\dots,m$.
Let $(c_1,\dots,c_n)$ be a non-trivial solution to \eqref{e:system} and assume that $\sum_{i=1}^n c_i\left(\alpha_i+\frac{\sigma_{ii}}2\right)\neq 0$. Then, for any starting densities $\bx(0)\in(0,\infty)^n$ with probability $1$
$$
\limsup_{t\to0}\dfrac{\ln\min\{x_1(t),\dots, x_n(t)\}}t<0.
$$

\end{thm}
\begin{proof}[Proof of Theorem \ref{t:CE_SDE}]
Suppose $g_i(\bx)=1$, $i=1,\dots, n$ and $\sum_jb_{ij}>b_m>0$ for any $i$ and some $b_m>0$.
Note that if $\sum_j b_{ij}=0$ then we can remove $R_j$ from the equation.
Assume that $$0<r_m\leq \liminf_{\|\bx\| \to\infty} \frac{|R_{j}(\bx)|}{|R_1(\bx)|}\leq r^M<\infty.$$
Then, since $\lim_{\|\bx\|\to\infty} R_j(\bx)=-\infty$, we have
when $|\bx|$ large that:
\begin{eqnarray*}
\sum_{i}\left(-x_i\alpha_i+x_i\sum_j b_{ij} R_j(\bx)\right)
&\leq& -\sum_{i}\left(x_i\alpha_i +x_i\sum_j b_{ij} |R_{j}(\bx)|\right)\\
&\leq& -\sum_{i}\left(x_i\alpha_i +x_i r_m \sum_j b_{ij} |R_{1}(\bx)|\right)\\
&\leq& - r_mb_m\left(\sum_i x_i\right)|R_1(\bx)|\\
&\leq& -\frac{r_m b_m}{mR_m}\left(\sum_i x_i\right)\sum_j |R_j(\bx)|
\end{eqnarray*}
which together with the linearity of the diffusion part implies that
 Assumption 1.1 from the work by {HN16} holds with $\bc=(1,\dots, 1)$.
As a result, for any starting point $\bx(0)\in(0,\infty)^n$ the SDE
\eqref{e:SDE} has a unique positive solution and by \cite{HN16} (equation (5.22)) with probability $1$
\begin{equation}\label{lnbx}
\limsup_{\|\bx\|\to \infty}\frac{\ln\|\bx\|}{t}\leq 0.
\end{equation}
By possibly replacing all $c_i$ by $-c_i$
we can assume that
$\sum_{i=1}^n c_i\left(\alpha_i+\frac{\sigma_{ii}}2\right)>0$.
Using this in conjunction with \eqref{e:SDE}, \eqref{e:system} and It\^{o}'s Lemma we see that
$$
\dfrac{\sum_{i=1}^n c_i\ln x_i(t)}{t}=-\sum_{i=1}^nc_i\left(\alpha_i+\frac{\sigma_{ii}}2\right)+\frac{1}{t}\sum_{i=1}^n c_i\int_0^tE_i(s)\,ds
$$
Letting $t\to\infty$ and using that $\lim_{t\to\infty}\dfrac {\sum_{i=1}^n c_iE_i(t)}{t}=0$ with probability $1$,  we obtain that with probability $1$
$$
\lim_{t\to\infty}\dfrac{\sum_{i=1}^n c_i\ln x_i(t)}{t}=-\sum_{i=1}^nc_i\left(\alpha_i+\frac{\sigma_{ii}}2\right)< 0.
$$
In view of \eqref{lnbx}
this implies that with probability $1$
$$
\limsup_{t\to\infty} \dfrac{\ln \left(\min\{x_1(t),\dots, x_n(t)\}\right)}{t}<0.
$$

\end{proof}

\section{Proof of Theorem \ref{t:CE_SDE3}}

\begin{thm}
Assume two species interact according to
\begin{equation*}
dx_i(t) = x_i(t)\left(-\alpha_i+b_iR(\bx(t))\right)\,dt + x_i(t)\sqrt{\beta_i x_i(t)} \,dB_i(t), ~i=1,2,
\end{equation*}
the resource $R$ depends linearly on the species densities
\[
R(\bx)=\bar R - a_1 x_1(t) - a_2x_2(t)
\]
and $b_i\bar R>\alpha_i, i=1,2$. Then there exist $\beta_1, \beta_2>0$ such that the two species coexist.
\end{thm}
\begin{proof}

Consider
\begin{equation*}\label{e:SDE2}
dx_i(t) = x_i(t)\left(-\alpha_i+b_iR(\bx(t))\right)\,dt + x_i(t)\sqrt{\beta_i x_i(t)} \,dB_i(t), ~i=1,2.
\end{equation*}
If the species $x_2$ is absent species $x_1$ has the one-dimensional dynamics
$$dx(t) = x(t)\left(-\alpha_1 + b_1 (\bar R-a_1x(t))\right)dt+ x(t)\sqrt{\beta_1 x(t)} \,dB_1(t).
$$
Since $b_1\bar R>\alpha_1$, we can use \cite{HN16} to show that
the process $x(t)$ has a unique invariant measure on $(0,\infty)$, say $\mu_1$.
Moreover, \cite[Lemma 2.1]{HN16} shows that
$$
\int_0^\infty\left(-\alpha_1+b_1(\bar R-a_1 x)-\beta_1x\right)\mu_1(dx)=0
$$
or
$$
\int_0^\infty x\mu_1(dx)=\dfrac{b_1\bar R-\alpha_1}{b_1a_1+\beta_1}.
$$
The invasion rate of $x_2$ with respect to $x_1$ can be computed by \eqref{e:lya} as
$$
\Lambda_{x_2}=\int_0^\infty \left(-\alpha_2+b_2(\bar R-a_1 x)\right)\mu_1(dx)
=(b_2\bar R-\alpha_2)- b_2a_1 \dfrac{b_1\bar R-\alpha_1}{b_1a_1+\beta_1}.
$$
Similarly, one can compute the invasion rate of $x_1$ with respect to $x_2$ as
$$
\Lambda_{x_1}=(b_1\bar R-\alpha_1)- b_1a_2 \dfrac{b_2\bar R-\alpha_2}{b_2a_2+\beta_2}.
$$
Since $b_i\bar R-\alpha_i>0$, one can easily see that
$
\Lambda_{x_2}>0$, and $\Lambda_{x_1}>0$
 if
$
\beta_1>\dfrac{b_2a_1(b_1\bar R-\alpha_1)}{b_2\bar R-\alpha_2}-b_1a_1.
$
and
$
\beta_2>\dfrac{b_1a_2(b_2\bar R-\alpha_2)}{b_1\bar R-\alpha_1}-b_2a_2.
$
If both invasion rates are positive we get by \cite{HN16} that the species coexist.

\end{proof}

\section{Proof of Theorem \ref{t:CE_SDE2}}\label{a:sde}
We construct an SDE example of two species competing for one abiotic resource and coexisting. We remark that this happens solely because of the random temporal environmental variation term.

\begin{thm}\label{CE_SDE2}
Suppose the dynamics of the two species is given by
\begin{equation}\label{e:3dSDE}
\begin{split}
dx_1(t)&=x_1(t)(-\alpha_1+f(\bar R-a_1x_1(t)-a_2x_2(t)))\,dt +\sigma_1x_1dB_1(t)\\
dx_2(t)&=x_2(t)(-\alpha_2+(\bar R-a_1x_1(t)-a_2x_2(t)))\,dt
\end{split}
\end{equation}
where $f$ is a continuously differentiable Lipschitz function satisfying $\lim_{x\to-\infty}f(x)=-\infty$,  $\dfrac{df(x)}{dx}>0, \dfrac{d^2f(x)}{dx^2}\leq 0$ for all $x\in\R$ and $\dfrac{d^2f(x)}{dx^2}<0$ for $x$ in some subinterval of $\left(-\infty,\frac{\bar R}{a_1}\right)$.
Let $a_1, a_2, \sigma_1, \alpha_1,\sigma_1, \bar R$ be any fixed positive constants satisfying $f(\bar R)>\alpha_1+\dfrac{\sigma_1^2}{2}$.
Then there exists an interval $(c_0,c_1)\subset (0,\infty)$ such that
the two species coexist for all $\alpha_2\in(c_0,c_1)$.
\end{thm}
\begin{proof}
The dynamics of species $x_1$ in the absence of species $x_2$ is given by the one-dimensional SDE
$$
dx(t)=x(t)(-\alpha_1+f(\bar R-a_1x(t)))\,dt +\sigma_1xdB_1(t).\\
$$
Since $\lim_{x\to\infty}f(\bar R-a_1x)=-\infty$,
and $f(\bar R)>\alpha_1+\dfrac{\sigma_1^2}{2}$,
this diffusion has a unique invariant probability measure $\mu$ on $(0,\infty)$ whose density is strictly positive on $(0,\infty)$ (see \cite{BS16} or \cite {MAO}).
Moreover, by noting that $\lim_{t\to\infty}\frac{\ln x(t)}{t}=0$ with probability $1$ (using Lemma 5.1 of \cite{HN16}) and using It\^{o}'s formula one sees that
\begin{equation}\label{sde-e1}
\int_0^\infty f(\bar R-a_1x)\mu(dx)=\alpha_1+\dfrac{\sigma_1^2}{2}.
\end{equation}
Since $f$ is a concave function and $\dfrac{d^2f(x)}{dx^2}>0$ for all $x$ in some subinterval of $\left(-\infty,\frac{\bar R}{a_1}\right)$
we must have by Jensen's inequality that
\begin{equation}\label{sde-e2}
\int_0^\infty f(\bar R-a_1x)\mu(dx)< f\left(\int_0^\infty (\bar R-a_1x)\mu(dx)\right).
\end{equation}
The fact that the function $f$ is strictly increasing together with \eqref{sde-e1} and \eqref{sde-e2} forces
\begin{equation}\label{sde-e3}
\eps_0:=\int_0^\infty (\bar R-a_1x)\mu(dx)- f^{-1}\left(\alpha_1+\dfrac{\sigma_1^2}{2}\right)>0,
\end{equation}
where $f^{-1}$ is the inverse of $f$ -- it exists because $f$ is strictly increasing. As a result, the invasion rate of species $x_2$ with respect to $x_1$, of the invariant probability measure $\mu$, can be computed using \eqref{e:lya} as
\[
\Lambda_{x_2}=-\alpha_2+\int_0^\infty (\bar R-a_1x)\mu(dx)=\eps_0+   f^{-1}\left(\alpha_1+\dfrac{\sigma_1^2}{2}\right) - \alpha_2.
\]
This implies that $\Lambda_{x_2}>0$ if and only if
\begin{equation}\label{e:ineq2}
 \alpha_2<\eps_0+   f^{-1}\left(\alpha_1+\dfrac{\sigma_1^2}{2}\right).
\end{equation}
The dynamics of species $x_2$ in the absence of species $x_1$ is
$$
dy(t)=y(t)(\bar R-\alpha_2-a_2y(t))dt.
$$
The positive solutions of this equation converge to the point $y^*=\dfrac{\bar R-\alpha_2}{a_2}$ if and only if
\begin{equation}\label{e:ineq}
\bar R>\alpha_2.
\end{equation}
The invasion rate of $x_1$ with respect to $x_2$ will be
\[
\Lambda_{x_1}=-\alpha_1-\frac{\sigma_1^2}{2}+f(\alpha_2).
\]
Note that since the function $f$ is increasing we get $\Lambda_{x_1}>0$ if and only if
\begin{equation}\label{e:ineq3}
\alpha_2>f^{-1}\left(\alpha_1+\frac{\sigma_1^2}{2}\right).
\end{equation}
Note that $f^{-1}\left(\alpha_1+\dfrac{\sigma_1^2}{2}\right)<\bar R$ since by assumption $f(\bar R)>\alpha_1+\dfrac{\sigma_1^2}{2}$.
As a result, making use of the inequalities \eqref{e:ineq2}, \eqref{e:ineq} and \eqref{e:ineq3} we get that $\Lambda_{x_2}>0, \Lambda_{x_1}>0$ if any only if
$$\alpha_2\in\left(f^{-1}\left(\alpha_1+\dfrac{\sigma_1^2}{2}\right), \bar R\wedge\left(f^{-1}\left(\alpha_1+\dfrac{\sigma_1^2}{2}\right)+\eps_0\right)\right).$$
This implies by Theorem \ref{t:pers} or by \cite{B18}[Theorem 4.4 and Definition 4.3] that the two species coexist.
\end{proof}

\section{Proof of Theorem \ref{t:CE_PDMP}}
\begin{thm}
Assume that
\begin{equation}\label{PDMP:bound}
\lim_{\|\bx\|\to\infty} \left(-\alpha_i(u)+\sum_{j=1}^mb_{ij}(u)R_j(\bx,u) \right)<0, i=1,\dots, n, u=1,\dots,N.
\end{equation}
Suppose further that there exists a vector $(c_1,\dots,c_n)$ that is simultaneously a solution to the systems \eqref{e:systems} for all $u\in\{1,\dots,N\}$. Then, with probability $1$,
$$
\limsup_{t\to0}\dfrac{\ln\min\{x_1(t),\dots, x_n(t)\}}t<0
$$
except possibly for the critical case when
\begin{equation}\label{e:crit}
\sum_{i=1}^n c_i\sum_{k=1}^N\alpha_i(k)\nu_k=0,
\end{equation}
where $(\nu_k)_{k\in\CN}$ is the invariant probability measure of the Markov chain $(r(t))$.
\end{thm}
\begin{proof}
Under the condition \eqref{PDMP:bound},
there exists an $M>0$ such that the set
$K_M:=\{\bx\in\R^n: \|\bx\|\leq M\}$
is a global attractor of \eqref{e:mix}.
As a result, the solution to \eqref{e:mix} eventually enters and never leaves the compact set $K_M$. In particular, this shows that the process $\bx(t)$ is bounded. Next, note that we can assume that
\begin{equation}\label{e:ineqq}
\sum_{i=1}^n c_i\sum_{k=1}^N\alpha_i(k)\pi_k>0.
\end{equation}
Otherwise, if $\sum_{i=1}^n c_i\sum_{k=1}^N\alpha_i(k)\pi_k<0$,
we can replace $c_i$ by $-c_i$, $i=1,\dots, n$ and then get \eqref{e:ineqq}.
Using \eqref{e1-pdm} and the fact that that $c_i$'s solve \eqref{e:systems} simultaneously we get
$$
\dfrac{\sum_{i=1}^n c_i\ln x_i(t)}{t}=-\dfrac1t\int_0^t\sum_{i=1}^n c_i\alpha_i(r(s))ds
$$
Letting $t\to\infty$ and using the ergodicity of the Markov chain $(r(t))$ we obtain that with probability $1$
$$
\limsup_{t\to\infty}\dfrac{\sum_{i=1}^n c_i\ln x_i(t)}{t}=-\liminf_{t\to\infty}\dfrac1t\int_0^t\sum_{i=1}^n c_i\alpha_i(r(s))ds=-\sum_{i=1}^n c_i\sum_{k=1}^N\alpha_i(k)\pi_k< 0 \text{ a.s.}
$$
Since $\bx(t)$ is bounded,
this implies that with probability $1$
$$
\limsup_{t\to\infty} \dfrac{\ln \min\{x_1(t),\dots, x_n(t)\}}{t}<0.
$$
\end{proof}

\section{Proof of Theorem \ref{t:CE_PDMP2}}
According to \cite{BL16, MZ16, MP16} it is enough to find an example for which the invasion rates $\Lambda_{x_1}, \Lambda_{x_2}$ are positive.
We will follow \cite{BL16} in order to compute the invasion rates of the two species.
Set for $u=1,2$ $\mu_u=-\alpha_1(u)+b_1(u)\bar R, \nu_u=-\alpha_2(u)+b_2(u)\bar R$
$\bar a_u=\dfrac{b_1(u)a_1(u)}{\mu_u}, \bar b_u=\dfrac{b_1(u)a_2(u)}{\mu_u}, \bar c_u=\dfrac{b_2(u)a_1(u)}{\nu_u}, \bar d(u)=\dfrac{b_2(u)a_2(u)}{\nu_u}$ ,
$p_u=\dfrac1{\bar a_u}, q_u=\dfrac1{\bar d_u}, \gamma_1=\dfrac{q_{12}}{\mu_u}, \gamma_2=\dfrac{q_{21}}{\nu_u}$.
If $p_1\ne p_2$, suppose without loss of generality that $p_1<p_2$.
Define the functions
$$\theta(x)=\dfrac{|x-p_1|^{\gamma_1-1}|p_2-x|^{\gamma_2-1}}{x^{1+\gamma_1+\gamma_2}}$$
and
$$P(x)=\dfrac{\bar a_2-\bar a_1}{|\bar a_2-\bar a_1|}\left[\dfrac{\nu_2}{\mu_2}(1-\bar c_2 x)(1-\bar a_1x)-\dfrac{\nu_1}{\mu_1}(1-\bar c_1 x)(1-\bar a_2x)\right].
$$
By \cite{BL16} we have
\begin{equation}\label{e:lambda}
\Lambda_{x_2}=\begin{cases}
\dfrac{1}{q_{12}+q_{21}}\left(q_{21}\nu_1(1-\bar c_1 p)+q_{12}\nu_2(1-\bar c_2p)\right)&\text{ if } p_1=p_2=p\\
p_1p_2\dfrac{\int_{p_1}^{p_2}\theta(x)P(x)dx}{\int_{p_1}^{p_2}\theta(x)dx}&\text{ if } p_1<p_2
\end{cases}
\end{equation}
The expression for $\Lambda_{x_1}$ can be obtained by swapping $\mu_i$ and $\nu_i$, $(\bar a_i,\bar c_i)$ with $(\bar d_i,\bar b_i)$, and $p_i$ with $q_i$.

For the example from Figures \ref{fig1} and \ref{figsw} we have used the integral equation from \eqref{e:lambda} together with the numerical integration package of Mathematica in order to find $\Lambda_{x_1}>0$ and  $\Lambda_{x_2}>0$.  This implies by Theorem \ref{t:pers}, \cite{BL16} or by \cite{B18}[Theorem 4.4 and Definition 4.3] that the two species coexist.

\end{document}